\documentclass[journal,12pt,draftclsnofoot,onecolumn]{IEEEtran}
\usepackage{cite}
\ifCLASSINFOpdf
\fi

\usepackage{amsfonts,amsthm} 
\usepackage[fleqn]{amsmath}
\usepackage{amssymb} 
\usepackage{bbm}
\usepackage{float}
\usepackage{graphicx}
\usepackage{algorithm}
\usepackage{algpseudocode}
\usepackage{subcaption}
\usepackage{siunitx}
\usepackage{tikz}
\usepackage{pgfplots}
\usepackage{epstopdf}
\usepackage{multirow}
\usepackage{multicol}
\usepackage[font=scriptsize,skip=3pt]{caption}
\usepackage{soul}
\usepackage{todonotes}
\usepackage{color, colortbl}
\usepackage{array}
\usetikzlibrary{patterns}
\usepackage{balance}

\usepackage[shortlabels]{enumitem}
\usepackage{footmisc}
\usepackage{balance}
\usepackage{mathtools, cuted}

\usepackage{pgfplots} 
\pgfplotsset{compat=newest} 
\pgfplotsset{plot coordinates/math parser=false} 
\newlength\figureheight 
\newlength\figurewidth

\usetikzlibrary{shapes.geometric}
\newcommand{\comment}[1]{}
\newcommand{\argmax}{\arg\!\max}

\graphicspath{{images/}}

\newtheorem{theorem}{Theorem}

\newtheorem{lemma}[theorem]{Lemma}
\newtheorem{proposition}[theorem]{Proposition}

\definecolor{Gray}{gray}{0.9}
\definecolor{LightCyan}{rgb}{0.88,1,1}

\algnewcommand{\IIf}[1]{\State\algorithmicif\ #1\ \algorithmicthen}
\algnewcommand{\IElse}[2]{\State\algorithmicelse\ #2\ }
\algnewcommand{\EndIIf}{\unskip\ \algorithmicend\ \algorithmicif}
\IEEEoverridecommandlockouts   
\begin{document}
	\setlength{\parskip}{0em}

\title{Hybrid NOMA-TDMA for Multiple Access Channels with Non-Ideal Batteries and Circuit Cost}
\author{Rajshekhar~Vishweshwar~Bhat,~\IEEEmembership{Graduate~Student~Member,~IEEE,}
	Mehul~Motani,~\IEEEmembership{Fellow,~IEEE,}
	and~Teng~Joon~Lim,~\IEEEmembership{Fellow,~IEEE}}
\maketitle
	\begin{abstract}
		We consider a multiple-access channel where the users are powered from batteries having non-negligible internal resistance. When power is drawn from the battery, a variable fraction of the power, which is a function of the power drawn from the battery, is lost across the internal resistance.  Hence, the power delivered to the load is less than the power drawn from the battery.  The users consume a constant power for the  circuit operation during transmission but do not consume any power when not transmitting. In this setting, we obtain the maximum sum-rates and achievable rate regions under various cases. We show that, unlike in the ideal battery case, the TDMA (time-division multiple access) strategy, wherein the users transmit orthogonally in time, may not always achieve the maximum sum-rate when the internal resistance is non-zero. The users may need to adopt  a hybrid NOMA-TDMA  strategy which combines the features of NOMA (non-orthogonal multiple access) and TDMA, 
		wherein a set of users are allocated fixed time windows for orthogonal single-user  and non-orthogonal joint transmissions, respectively. We also numerically show that the maximum achievable rate regions in NOMA and TDMA strategies are contained within the maximum achievable rate region of the hybrid NOMA-TDMA strategy.

	\end{abstract}
	\IEEEpeerreviewmaketitle
	\section{Introduction}
	In battery-powered communication devices, usually, batteries are designed to have high energy density so that a large amount of energy can be stored in a small volume. Quite often, such {high energy-density} batteries exhibit a low power density, i.e., the batteries cannot deliver a large amount of energy in a short period of time with high efficiency \cite{EDvsPD}. This limitation can be abstractly accounted for by modeling the battery as a voltage source with a series internal resistance. When the internal resistance is non-negligible, a variable fraction of the power drawn\footnote{The `power drawn from the battery' or equivalently, the `discharge power' is the rate at which energy is depleted from the battery, internally.} from the battery is lost across the internal resistance thereby reducing the power delivered to the load.

	In this work, we consider a multiple access channel (MAC) with the users having batteries with non-negligible internal resistance. The users are switched on only when they transmit data and they consume a constant power, referred to as the circuit cost, when switched on. Over the remaining period of time, the users go to the sleep state where they consume a negligible amount of power. When the circuit cost is non-zero, to avoid energy wastage in the circuit,  it may be optimal to transmit for a small amount of time, in bursts \cite{MAC-PC,GP}. However, transmitting in bursts entails  \emph{high-rate} discharge of the battery. In practice, when  the  internal resistance is non-zero,  the battery output current collapses if the load attempts to draw too much power from the battery \cite{IR,Krieger}. Hence, it may be inefficient, or in some cases, infeasible to transmit in bursts. 
	This results in a trade-off between the losses in the circuit and the internal resistance. 
	In this work, we address this trade-off by jointly optimizing the transmit power, discharge power of the battery  and the duration of transmission.  
	
	The impact of circuit cost on communication rates over a point-to-point channel and a MAC has been studied in \cite{MAC-PC,GP}, when the transmitters are powered from ideal batteries. 
	The impact of circuit cost has also been studied for interference channels \cite{IC-PC,ZIC}. In \cite{BC-PC}, a broadcast channel with circuit cost has been considered. Under the assumption that on-off states of the users are not used for signaling,  the authors in \cite{MAC-PC,GP} show that when the circuit cost is large, bursty transmission achieves the capacity in the point-to-point channel. For a two-user MAC, it has been shown that the TDMA strategy, wherein the users transmit orthogonally in time, achieves the maximum sum-rate. Further, a strategy which we refer to as the hybrid {NOMA-TDMA}, wherein the users transmit simultaneously for a fraction of time, in addition to individual transmission in disjoint intervals, can achieve any rate pair in the maximum achievable rate region. In all the above works, an important takeaway is that when the circuit cost is non-zero, transmitting all the time, as in the \emph{zero-cost} case is no longer optimal \cite{IC-PC}.  
	
	In this work, we account for the internal resistance of the batteries, in addition to the circuit cost. As in \cite{GP,IC-PC,BC-PC}, our focus is to allocate the degrees of freedom in terms of the transmission duration and the power, under the assumption that on-off states do not carry any information.
	The main contributions of the paper are:
	\begin{itemize}[leftmargin=*]
		\item We show that, unlike in the ideal battery case, the TDMA strategy does not always achieve the maximum sum-rate when the internal resistance is non-zero.
		\item We then obtain maximum sum-rates and maximum achievable rate regions using a hybrid NOMA-TDMA strategy.  
		\item By numerical simulations, we show that the maximum achievable rate regions in NOMA and TDMA strategies are contained within the maximum achievable rate region of the hybrid NOMA-TDMA strategy.
	\end{itemize} 
	The remainder of the paper is organized as follows. The system model is  presented in Section \ref{sec:model}. We first study the single user case in Section \ref{sec:1user}. We then study the two-user case in Section \ref{sec:2users} and generalize the results to an arbitrary number of users in Section \ref{sec:Uusers}. Numerical results are presented in Section \ref{sec:numerical results}, followed by concluding remarks in Section \ref{sec:reflections}.
	
	\section{System Model and Assumptions}\label{sec:model}
	In this work, we consider a MAC with $U$ users indexed by $1,\ldots,U$, transmitting data to an access point. We assume that each of the users have an infinite backlog of data. The circuit cost of user $u\in \mathcal{U}\triangleq \{1,\ldots,U\}$ is $\gamma^{(u)}$ \si{\watt}, i.e., the user $u$ consumes $\gamma^{(u)}$ \si{\watt}  for circuit operation during transmission but does not consume any power when not transmitting. 
	We assume the user $u\in \mathcal{U}$ has $B^{(u)}$ \si{\joule}  in its battery. Based on practical batteries, we consider the following 
	non-linear battery discharge model. 
	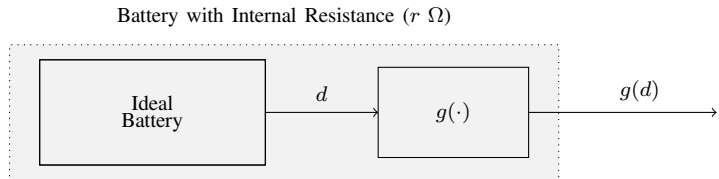
\begin{figure}[t]
		\centering
		\begin{tikzpicture}[scale=2]
		\draw [dotted,fill=black!5](3.55,-0.45) rectangle (7.2,0.45) node [above, midway,yshift=1cm] {\scriptsize Battery with Internal Resistance ($r$ \si{\ohm})};   
		\draw (3.75,-0.35) rectangle (5.25,0.35) node [pos=0.5,yshift=0.13cm] {\scriptsize Ideal }; 
		\draw (3.75,-0.35) rectangle (5.25,0.35) node [pos=0.5,yshift=-0.13cm] {\scriptsize Battery}; 
		\draw [->] (5.25,0) -- (6,0) node [midway,above] {\scriptsize $d$};
		\draw (6,-0.3) rectangle (7,0.3) node [pos=0.5] {\scriptsize $g(\cdot)$};
		\draw [->] (7,0) -- (8.25,0) node [midway,above,xshift=0.23cm] {\scriptsize $g(d)$}; 
		\end{tikzpicture}
		\caption{The battery with internal resistance is depicted as an ideal battery with an additional block that models the effect of the internal resistance. When the battery is discharged at $d$ \si{\watt},  the rate at which energy is available at the load is $g(d)$ \si{\watt}.}
		\label{fig:flow-diagram}
	\end{figure}
	In practice, when the transmitter attempts to discharge a battery, a fraction of the  discharging power is lost in the form of heat dissipated by the internal resistance of the battery. To describe this impact of the internal resistance, we present a block diagram in Fig. \ref{fig:flow-diagram}, where the battery with internal resistance is depicted as an ideal battery with an additional block that models the effect of the internal resistance. When the battery is discharged at $d$ \si{\watt},  the rate at which energy is available at the load (transmitter) is $g(d)$ \si{\watt} and the remaining  $d-g(d)$ \si{\watt} is lost in the internal resistance.
	Based on  \cite{IR}, we assume that  $g(d)$ is a concave function of $d$, and $g(d)\leq d$, for a fixed internal resistance, $r$.
	We represent $g(\cdot)$ in user $u$ as $g^{(u)}(\cdot)$ in the rest of the work. 
	
	
	We assume the transmission takes place over an additive white Gaussian noise channel of unit power spectral density. 
	When a user transmits with a constant power $P$ \si{\watt}, the maximum achievable rate is assumed to be $\log(1+P)$ nats/second. We also assume  the perfect synchronization among the users, as in \cite{GP,IC-PC,ZIC,BC-PC}. 

	\section{Single-User Case }\label{sec:1user}
	We now consider the single user case. We obtain the maximum achievable sum-rate and highlight the impact of the initial energy stored in the battery on the optimal transmit duration. We drop the user index for simplicity. Since the circuit cost can be non-zero, it may not be optimal to transmit over the entire $T$ seconds \cite{GP}. Hence, we assume the transmission takes place over $\tau$ $(\tau\leq T)$ seconds. 
	Let the battery discharge power be $d$ \si{\watt}. Then, the power available  at the transmitter is $g(d)$ \si{\watt}. Let $P$ \si{\watt} be the transmit power. Then, recalling that the circuit cost is $\gamma$ \si{\watt}, the power consumed at the transmitter is $(P+\gamma)$ \si{\watt}. This power must be less than or equal to the available power, $g(d)$ \si{\watt}, i.e.,
	$P+\gamma\leq g(d)$ must be satisfied. Further, the total energy drawn from the battery, $\tau d$, must be less than $B$, i.e.,  $\tau d\leq B$ must be satisfied. 
	To maximize the total number of bits transmitted over $[0,T]$, we thus need to solve the following problem. 
	\begin{subequations}\label{eq:p2p-1}
		\begin{align}
			\mathrm{(P1)}:\underset{\substack{P,d,\tau}}{\text{maximize}} &\;\;\tau\log(1+P)\;\;\;&&\\
			\text{subject to} &\;\; P\leq g(d)-\gamma\label{eq:p2p-1-c1} &&\\
			& \;\;\tau d\leq B\label{eq:p2p-1-c2} && \\
			& \;\; 0\leq \tau \leq T,\; P \geq 0\label{eq:p2p-1-c4}
		\end{align}
	\end{subequations}
	Clearly, the objective function and \eqref{eq:p2p-1-c2} are non-convex. Based on certain observations, we now propose an equivalent problem in the following lemma. 
	\begin{lemma}\label{lemma:p2p0}
		Let $D_{0}\triangleq \argmax_{d} g(d)$. Consider the following convex optimization problem. 
		\begin{align}\label{eq:p2p-3}
			&\mathrm{(P2)}:\underset{\substack{B/D_0\leq \tau \leq T}}{\text{maximize}} \;\tau\log\left(1+g\left( \frac{B}{\tau}\right)-\gamma\right)&&
		\end{align}
		Now, $\mathrm{(P2)}$ in \eqref{eq:p2p-3} is equivalent to $\mathrm{(P1)}$ in \eqref{eq:p2p-1}. 
	\end{lemma}
	\begin{proof}
		See Appendix A.  
	\end{proof}
	%
	%
	%
	%
	From the above lemma, we note that the maximum power that the battery can deliver is $D_0$ \si{\watt} and when $g(D_0)<\gamma$, the solution is infeasible as the battery cannot deliver sufficient power even to run the circuitry. Further, when feasible, $\mathrm{(P2)}$ can be solved by bisection search efficiently.  
	We now study the impact of variation of $B$ on the solution of $\mathrm{(P2)}$ using the Karush-Kuhn-Tucker (KKT) conditions. 
	\begin{proposition}\label{lemma:p2p}
In the optimal solution to 	$\mathrm{(P2)}$, the optimal transmit duration, $\tau^*$ is a linearly increasing function of  $B$ in the range $(B/D_0, T)$. 
\end{proposition}
	\begin{proof}
		See Appendix B. 
	\end{proof}
From the above result, we note that as $B$ increases, the optimal discharge power, $d^*=B/\tau^*$ remains constant when $B/D_0<\tau^*<T$.
We also note the above result is similar to the optimal result when the battery is ideal, with $g(d)=d$,  where the optimal transmit duration is a  linearly increasing function of $B$ over $(0,T)$  \cite{GP}.  Further, for some $B'$,  if the optimal transmit duration, $\tau^*=T$, then  for any $B>B'$, we have,  $\tau^*=T$, due to the constraint that $\tau\leq T$.

	\section{Two-User Multiple Access Channel}\label{sec:2users}
	We now consider the two-user MAC. We first propose the optimal frame structure and obtain the maximum achievable rate region. We then obtain and compare maximum sum-rates under various strategies,  including the optimal strategy and  strategies which are optimal under certain special assumptions. 
	
	
	\subsubsection{Optimal Frame Structure}
	As in the single-user case, since the circuit cost can be non-zero, it may not be optimal for both the users to transmit over the entire  time duration of $T$ seconds. In the two-user MAC, at a given instant of time, none, one or both of the users may be transmitting. To allow all the above degrees of freedom, we adopt the communication frame structure shown in  Fig. \ref{fig:MAC-phi}. The frame is divided into four phases of lengths $\tau_1,\tau_2,\tau_3$ and $\tau_4$ in which none of the users, only user $1$, only user  $2$ and both the users transmit, respectively. 
	We represent the transmit and discharge powers in phase $i\in \{1,\ldots,4\}$ in user $u\in \{1,2\}$ by   $P^{(u)}_i$ and $d^{(u)}_i$, respectively. 
	Clearly, for any given length of the phases, the order in which the phases are transmitted does not affect the feasible range of power allocations, as the entire energy is available at the start of the transmission. Note that in phase $4$, the users superimpose their codewords and transmit, i.e., they adopt the NOMA strategy. The information is decoded at the receiver by successive interference cancellation. Also note that the phases are time multiplexed. Hence, we refer to a strategy that uses the frame structure  in Fig. \ref{fig:MAC-phi} as the hybrid NOMA-TDMA strategy. 
	We also note that such a frame structure has  been considered in \cite{GP} for the two-user MAC with ideal batteries to obtain  the maximum achievable rate region.

	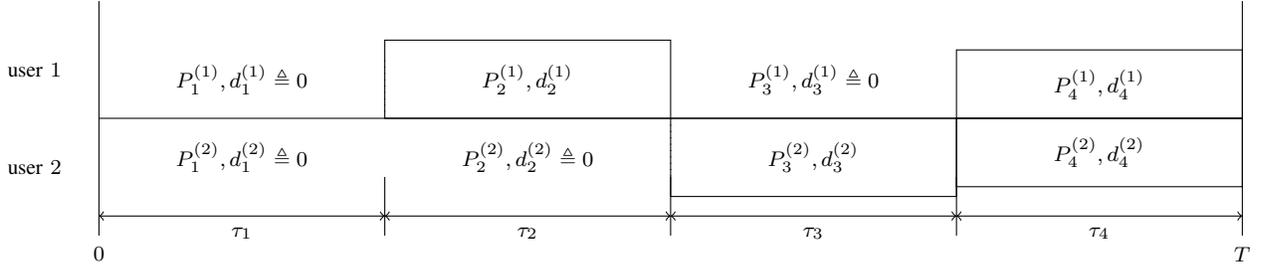
\begin{figure}[t]
		\centering
		\begin{tikzpicture} [scale=2,yscale=1.3,xscale=.95]
		\node at (-0.45,0.25) {\scriptsize user $1$};
		\node at (-0.45,-0.25) {\scriptsize user $2$};
		\draw (0,0)  -- (8,0);
		\draw (0,0.6) -- (0,-0.6) node [below] {\scriptsize $0$};
		\draw (8,0.6) -- (8,-0.6) node [below] {\scriptsize $T$};
		
		\draw [opacity=0](0,0) rectangle (2,-0.4) node [pos=0.5,opacity=1] {\scriptsize $P^{(2)}_1,d^{(2)}_1\triangleq 0$};
		\draw [opacity=0](0,0) rectangle (2,0.4) node [pos=0.5,opacity=1] {\scriptsize $P^{(1)}_1, d^{(1)}_1\triangleq 0$};
		
		\draw (2,0) rectangle (4,0.4) node [pos=0.5] {\scriptsize $P^{(1)}_2,d^{(1)}_2$};
		\draw [opacity=0](2,0) rectangle (4,-0.4) node [pos=0.5,opacity=1] {\scriptsize $P^{(2)}_2, d^{(2)}_2\triangleq 0$};
		
		\draw (4,-0.4) rectangle (6,0) node [pos=0.5] {\scriptsize $P^{(2)}_3,d^{(2)}_3$}; 
		\draw [opacity=0](4,0) rectangle (6,0.4) node [pos=0.5,opacity=1] {\scriptsize $P^{(1)}_3, d^{(1)}_3 \triangleq 0$};	
		
		\draw (6,0) rectangle (8,0.35) node [pos=0.5] {\scriptsize $P_4^{(1)},d_4^{(1)}$};   
		\draw (6,-0.35) rectangle (8,0) node [pos=0.5] {\scriptsize $P_4^{(2)},d_4^{(2)}$};

		\draw [dotted](2,0) -- (2,0.3);
		\draw [dotted](4,0) -- (4,-0.3); 
		\draw (2,-0.3) -- (2,-0.6);       
		\draw (4,-0.3) -- (4,-0.6);
		\draw (6,-0.3) -- (6,-0.6);   
		\draw [<->](0,-0.5) -- (2,-0.5) node [below, midway] {\scriptsize $\tau_1$};
		\draw [<->](2,-0.5) -- (4,-0.5) node [below, midway] {\scriptsize $\tau_2$};
		\draw [<->](4,-0.5) -- (6,-0.5) node [below, midway] {\scriptsize $\tau_3$};   
		\draw [<->](6,-0.5) -- (8,-0.5) node [below, midway] {\scriptsize $\tau_{4}$};       
		\end{tikzpicture}
		\caption{\scriptsize Transmission frame structure adopted in the work for the two-user  case. }
		\label{fig:MAC-phi} 
	\end{figure}
	\subsubsection{Maximum Achievable Rate Region}
	Let $\gamma_i^{(u)}\in\{\gamma^{(u)},0\}$ be the constant power, required to run the circuitry,  in phase $i$, where $\gamma_i^{(u)}=\gamma^{(u)}$ if user $u$ transmits in phase $i$ and  $\gamma_i^{(u)}=0$ otherwise. Since none of the users transmit in phase $1$ and user $2$ (user $1$) does not transmit in phase $2$ (phase $3$), we apply the following constraints.
	\begin{align}
		&\gamma^{(1)}_1, \gamma^{(2)}_1, \gamma^{(2)}_2, \gamma^{(1)}_3, d^{(1)}_1, d^{(2)}_1 , d^{(2)}_2, d^{(1)}_3\triangleq 0\label{eq:mac_c1}\\
		&P^{(1)}_1, P^{(2)}_1, P^{(2)}_2, P^{(1)}_3  \triangleq 0	\label{eq:mac_c1-1}
	\end{align}  
	Further, as in \eqref{eq:p2p-1-c1}-\eqref{eq:p2p-1-c4},  the  following constraints must be satisfied. 
	\begin{align}
		&P_i^{(u)} \leq g^{(u)}\left( d^{(u)}_i\right)-\gamma^{(u)}_i, \;\; P_i^{(u)}\geq 0,\; \sum_{i=1}^{4}\tau_i \leq T \label{eq:mac_c2}&& \\
		&\sum_{i=1}^{4}\tau_id^{(u)}_i\leq B^{(u)},\; \tau_i\geq 0,\;0\leq d^{(u)}_i\leq D_{0}^{(u)}\label{eq:mac_c5}&&
	\end{align}
	for all $i= 1,\ldots, 4$ and $u=1,2$, where $D_{0}^{(u)}\triangleq \argmax_{d} g^{(u)}(d)$. We apply  $d^{(u)}_i\leq D_{0}^{(u)}$ because all effective discharge powers $g^{(u)}(d^{(u)}_i)$ are obtainable using this range of $d^{(u)}_i$. 
	The rest of the constraints should be self-explanatory. 
	Now, in phase $i=2,3$, the maximum achievable rate of user $u$ is $\tau_i\log(1+P_i^{(u)})$. In phase $4$, the maximum achievable rate of user $u$ is $\tau_4\log(1+P_4^{(u)})$ and the maximum achievable sum-rate is $\tau_4\log(1+P_4^{(1)}+P_4^{(2)})$ \cite{Cover}.  Hence, the maximum achievable rate region is given by the convex hull of the closure of all $(R^{(1)}, R^{(2)})$ satisfying
	\begin{align}
		&R^{(u)}\leq \tau_i\log(1+P_i^{(u)})+\tau_{4}\log(1+P_4^{(u)})&&\label{eq:cap_region1}\\
		&R^{(1)}+R^{(2)}\leq \tau_2\log(1+P_2^{(1)})+\tau_3\log(1+P_3^{(2)})+\tau_4\log(1+P_4^{(1)}+P_4^{(2)})&& \label{eq:cap_region2}
	\end{align}
	for $i=2,3, \;u=1,2$ subject to \eqref{eq:mac_c1}-\eqref{eq:mac_c5}. 
	\subsubsection{Maximum Sum-Rates}
	We now obtain and compare the maximum sum-rates under various strategies.

	\paragraph{ NOMA }
	In this strategy, both the users transmit for the entire time duration, i.e., $\tau_1, \tau_2, \tau_3=0,\tau_4=T$,  with a constant power. Since it is optimal to exhaust the battery subject to the maximum discharge power constraint, $d^{(u)}_i\leq D_{0}^{(u)}$, we have,  $d_4^{(u)*} =\min(B^{(u)}/T,D_{0}^{(u)})$ and $P_4^{(u)*}=g^{(u)}(d_4^{(u)*})-\gamma^{(u)}$ for $u=1,2$.
	The maximum sum-rate is given by
	\begin{align}\label{eq:cap_regionS12}
		R_{\rm NOMA}\triangleq T\log\left(1+\sum_{u=1}^{2}\left(g^{(u)}\left(d_4^{(u*)}\right)-\gamma^{(u)}\right)\right)
	\end{align}
	

	\paragraph{TDMA}
	In this strategy, the users never transmit simultaneously, i.e., $\tau_1,\tau_2,\tau_3\geq 0,\tau_4=0$. The user $1$ transmits for $\tau_2$ seconds, while the user $2$ transmits for  $\tau_3$ seconds, subject to $\tau_2+\tau_3\leq T$.   Hence, there is no interference between the signals and the decoding is precisely as in the single-user case.  The battery at user $u$ is discharged at $\min(B^{(u)}/\tau_i,D_{0}^{(u)})$ \si{\watt}, as it is optimal to consume the entire energy available. 
	Clearly, the maximum sum-rate in this case is given by
	\begin{subequations}\label{eq:mac-II}
		\begin{align}\label{eq:MAC-ZI}
			\mathrm{(P3)}:R_{\mathrm{TDMA}}\triangleq\underset{\substack{ \tau_2,\tau_3}}{\text{maximize}}& \;R^{(1)}+R^{(2)}&&\\
			\text{subject to} & \; \tau_2+\tau_3\leq T,\; \tau_i\geq 0&&
		\end{align}
	\end{subequations}
	for $i=2,3$, 
	where $R^{(1)}+R^{(2)}$ is obtained from \eqref{eq:cap_region2} with equality.


	\paragraph{Hybrid NOMA-TDMA}
	This is the most general case with $\tau_1,\tau_2,\tau_3,\tau_4 \geq  0$.  To maximize the sum-rate, we need to solve the following optimization problem. 
	\begin{subequations}\label{eq:MAC}
		\begin{align}
			\mathrm{(P4)}:\underset{\substack{P_i^{(u)},d^{(u)}_i, \tau_i \\ \{i\}_1^4,u\in \{1,2\}}}{\text{maximize}} &\;\;\sum_{i=1}^{4}\tau_i\log\left(1+P^{(1)}_i+P^{(2)}_i\right)&&\\
			\text{subject to} & \;\;\;\;\eqref{eq:mac_c1}-\eqref{eq:mac_c5}&&
		\end{align}
	\end{subequations}
	where the objective function is obtained from \eqref{eq:cap_region2}. 
	Note that $\mathrm{(P4)}$ in \eqref{eq:MAC} is non-convex. 
	We now transform the problem to a convex problem  by a change of variables. Define $E_i^{(u)}\triangleq \tau_iP_i^{(u)}$ and $e_i^{(u)}\triangleq \tau_id^{(u)}_i$  for all  $u\in \{1,2\},i\in \{1,\ldots,4\}$.  Now, $\mathrm{(P4)}$ in \eqref{eq:MAC} can be transformed to, 
	\begin{subequations}\label{eq:p_MAC1}
		\begin{align}
			\mathrm{(P5)}:\underset{\substack{E_i^{(u)},e_i^{(u)},\tau_i\\  \{i\}_1^4, u\in \{1,2\}}}{\text{maximize}} &\;\; \sum_{i=1}^{4}\tau_i\log\left(1+\frac{E^{(1)}_i+E^{(2)}_i}{\tau_i}\right)&&\\
			\text{subject to} & \;\;E_i^{(u)}\leq \tau_ig^{(u)}\left( \frac{e_i^{(u)}}{\tau_i}\right)-\tau_i\gamma_i^{(u)}\label{eq:p_nonEH1_c1}&& \\
			&\;\; \sum_{i=1}^{4}e_i^{(u)}\leq B^{(u)}&&\\
			&\;\; E_i^{(u)}, \tau_i \geq 0,\sum_{i=1}^{4}\tau_i\leq T&&\\
			&\;\; 0\leq e_i^{(u)}\leq \tau_iD_{0}^{(u)} &&\label{eq:p_nonEH1_c4}&&\\
			&\;\;  \eqref{eq:mac_c1}, E^{(1)}_1,E^{(2)}_1,E^{(2)}_2,E^{(1)}_3=0\label{eq:p_nonEH1_c5}&
		\end{align}
	\end{subequations}
	for all $i\in\{1,\ldots,4\}$ and $u=1,2$, where we have multiplied \eqref{eq:mac_c2} by $\tau_i$ to obtain \eqref{eq:p_nonEH1_c1}. Noting that the perspective preserves convexity \cite{EH-PC1}, we recognize that $\mathrm{(P5)}$ in \eqref{eq:p_MAC1} is a convex optimization problem and  it can be solved using standard numerical techniques. We represent the optimal objective value of $\mathrm{(P5)}$ in \eqref{eq:p_MAC1} by $R_{\rm NOMA-TDMA}$.

	\subsubsection{Comparison of the  Strategies}
	We first note that when the battery is ideal, only TDMA achieves the maximum sum-rate if the circuit cost is non-zero, i.e., $R_{\mathrm{TDMA}}> R_{\rm NOMA}$,  and both NOMA and the TDMA achieve the maximum sum-rate when the circuit cost is zero, i.e.,  $R_{\mathrm{TDMA}}=R_{\rm NOMA}$. When the battery is non-ideal, we have the following lemma. 
	\begin{lemma}\label{thm:sum-rate}
		When the discharging functions, $g^{(u)}(\cdot), u\in \{1,2\}$ are strictly concave,   
		\begin{enumerate}[leftmargin=*]
			\item $R_{\mathrm{TDMA}}$ in \eqref{eq:mac-II} may not be always greater than or equal to $R_{\mathrm{NOMA}}$ in \eqref{eq:cap_regionS12}. 
			\item The TDMA strategy does not always achieve the maximum sum-rate, $R_{\rm NOMA-TDMA}$. 
		\end{enumerate}
	\end{lemma}
	\begin{proof}
		See Appendix C.
	\end{proof}
	In the example constructed in the proof,  $R_{\rm NOMA}=R_{\rm NOMA-TDMA}$ and the second statement directly follows from the first. In the numerical results, we construct an  example where $R_{\rm TDMA}<R_{\rm NOMA}<R_{\rm NOMA-TDMA}$. 
	We now note that in the NOMA strategy, both the users transmit all the time and in the TDMA strategy, the users transmit only in disjoint intervals.  
	Hence, the overall circuit energy consumed in the  NOMA strategy is higher than the TDMA strategy.  On the other hand, since the users transmit for a shorter duration of time in the TDMA strategy as compared to the  NOMA strategy, the transmit powers in the TDMA strategy are always higher than the  NOMA strategy. Hence, the losses due to the internal resistance is higher in the TDMA strategy than the  NOMA strategy. In  summary, the TDMA strategy reduces the loss in the circuit at the cost of increased loss across the internal resistance and, the  NOMA strategy, reduces the loss across the internal resistance at the cost of  increased loss in the circuit. The hybrid NOMA-TDMA optimally trades off between the losses in the circuit and the internal resistance.

\subsection{Plotting the maximum achievable rate region}
We now obtain the corner points on the maximum achievable rate region  given by \eqref{eq:cap_region1}-\eqref{eq:cap_region2}. We illustrate some important corner points in Fig. \ref{fig:MAC-cap}. 
\subsubsection{Points on the maximum-sum rate line, $BC$}
Let $E_i^{(u)*}, e_i^{(u)*}, \tau_i^*$ be the optimal solution to $\mathrm{(P5)}$ in \eqref{eq:p_MAC1}. Then, any point on the maximum sum-rate line, $BC$ (see Fig. \ref{fig:MAC-cap}) can be achieved by changing the decoding order and varying the time-sharing factor, as follows.  In the $4^{th}$ phase, for $\alpha\tau_4^*$ seconds, where $\alpha\in [0,1]$, we decode the user $1$ first and then cancel the interference to decode user $2$, and for the remaining $(1-\alpha)\tau_4^*$ seconds, we change the decoding order. Hence, for any given $\alpha$, we achieve the following individual rates on the maximum sum-rate line. 
\begin{align}
R^{(1)}(\alpha)=
&\tau_2^* \log\left(1+\frac{E_2^{(1)*}}{\tau_2^*}\right)+\alpha\tau_4^*\log\left(1+\frac{{E}_4^{(1)*}}{\tau_4^*+{E}_4^{(2)*}}\right)\nonumber&&\\
&\qquad\qquad+(1-\alpha)\tau_4^*\log\left(1+\frac{{E}_4^{(1)*}}{\tau_4^*}\right)&&\\
R^{(2)}(\alpha)=&\tau_3^*\log\left(1+\frac{{E}_3^{(2)*}}{\tau_3^*}\right)+\alpha\tau_4^*\log\left(1+\frac{{E}_4^{(2)*}}{\tau_4^*}\right)\nonumber&&\\
&\qquad+(1-\alpha)\tau_4^*\log\left(1+\frac{{E}_4^{(2)*}}{\tau_4^*+{E}_4^{(1)*}}\right)&&
\end{align}
When $\alpha=1$ and $\alpha=0$, we achieve points $B$  and $C$, respectively.  
\subsubsection{Maximum achievable rate of an user when the other user transmits at its maximum achievable rate}

\begin{figure}[t]
	\centering
	\begin{tikzpicture}[scale=1.9]
	\draw [<->,thick] (0,2.3) node (yaxis) [above] {\scriptsize $R^{(2)}$}
	|- (2.3,0) node (xaxis) [right] {\scriptsize  $R^{(1)}$};
	
	\draw (0,2.2) node [left]{\scriptsize $C^{(2)}$ }-- (1.25,2.2) ;
	\draw (2.2,0) node [below,xshift=0.1cm]{\scriptsize $C^{(1)}$ } -- (2.2,1.25) ;
	\draw [blue](1.55,2.075) --(2.075,1.55)  node[midway,rotate=-45, xshift=0.5cm, yshift=0.6cm]{\scriptsize max sum-rate};
	\draw [->] (1.95,1.9) -- (1.75,1.9);
	\draw [dashed] (1.25,2.2) -- (1.25,0) node[below] {\scriptsize $R^{(1),C^{(2)}}$} ;
	\draw [dashed]  (2.075,1.55) -- (2.075,0) node[left,yshift=-0.25cm,xshift=0.2cm] {\scriptsize $R^{(1),C}$} ;
	\draw [dashed] (2.2,1.25) -- (0,1.25) node[left] {\scriptsize $R^{(2),C^{(1)}}$} ; 
	\draw [red](1.25,2.2) to[bend left] (1.55,2.075);
	\draw [red](2.2,1.25) to[bend right] (2.075,1.55);
	\fill (1.55,2.075) circle [radius=1pt] node[right]{\scriptsize $B$ } ;
	\fill (2.075,1.55) circle [radius=1pt] node[right]{\scriptsize $C$ } ;
	\fill (1.25,2.2) circle [radius=1pt] node[above] {\scriptsize $A$} ;
	\fill (2.2,1.25) circle [radius=1pt] node[right] {\scriptsize $D$} ;
	\end{tikzpicture}
	\caption{An illustration of the maximum achievable rate region of the generalized-TDMA strategy.  All the points on the $BC$ segment are the maximum sum-rate points. }
	\label{fig:MAC-cap}
\end{figure}
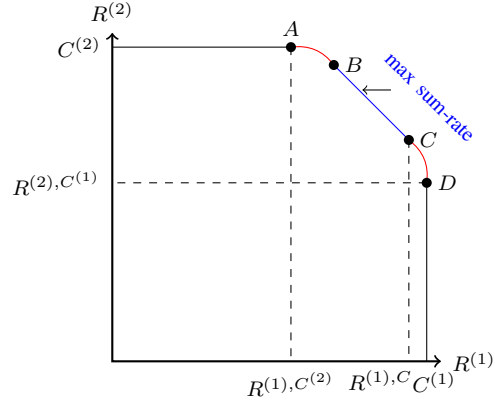

The maximum achievable rate in user $1$ (user $2$)  when user $2$ (user $1$) transmits at its maximum achievable rate is specified by the point $A$ ($D$) in Fig. \ref{fig:MAC-cap}. 
We now obtain the optimal value of the optimization variables to achieve points $A$ and $D$.
Without loss of generality, consider the case when user $2$ transmits at its maximum achievable rate, $C^{(2)}$. The maximum achievable rate is achieved when user $1$ is decoded first and the interference from user $1$ is canceled out before decoding the user $2$. The optimal solution can be obtained by solving $\mathrm{(P2)}$ in \eqref{eq:p2p-3} for user $2$. Let $\tau_3^{(2),C^{(2)}}$ be the optimal transmission duration  in this case, where the superscript specifies that the user $2$ transmits at its maximum achievable rate. To achieve the maximum achievable rate, the user $2$ transmits at the constant power given by $P_3^{(2),C^{(2)}}=g^{(2)}(B^{(2)}/\tau_3^{(2),C^{(2)}})-\gamma^{(2)}$ over $\tau_3^{(2),C^{(2)}}$ seconds. Then, the achievable rate at user $1$ is,
\begin{align}\label{eq:rate1}
R^{(1),C^{(2)}}=\tau_2\log\left(1+\frac{{E}_2^{(1)}}{\tau_2}\right)+\tau_4\log\left(1+\frac{{E}_4^{(1)}/\tau_4}{1+P_3^{(2),C^{(2)}}}\right)&&
\end{align}
where we have noted that when the user $1$ transmits in the same interval as user $2$, the signal from user $2$ acts as the noise  for user $1$.   Recalling that $\tau_2$ is the duration over which only user $1$ transmits, we have the following constraints.  
\begin{align}\label{eq:01}
\tau_2\leq T-\tau_3^{(2),C^{(2)}},\;\; \tau_4\leq  \tau_3^{(2),C^{(2)}}
\end{align}
Now, the optimal rate in user $1$ when user $2$ transmits at its maximum achievable rate can be obtained by solving the following  problem. 
\begin{subequations}\label{eq:p_MAC2}
	\begin{align}
	\underset{\substack{{E}_2^{(1)},{E}_{4}^{(1)},{e}_2^{(1)},{e}_{4}^{(1)},\\ \tau_2,\tau_4}}{\text{maximize}} &\;\;R^{(1),C^{(2)}} &&\\
	\text{subject to}&\;\;\eqref{eq:p_nonEH1_c1}-\eqref{eq:p_nonEH1_c5}, \eqref{eq:01}
	\end{align}
\end{subequations}
where \eqref{eq:p_nonEH1_c1}-\eqref{eq:p_nonEH1_c5} are considered only for $u=1$. 
Along the above lines, we can easily find the maximum rate at which user $2$ can transmit when the user $1$ transmits at its maximum achievable rate, $C^{(1)}$. 

\subsubsection{Obtaining curves $AB$ and $CD$ in  Fig. \ref{fig:MAC-cap}}
When the internal resistance is zero, there is a unique maximum sum-rate point that is achieved when $\tau_4=0$ and $\tau_2+\tau_3=T$ \cite{GP}. In such a case, the points $B$ and $C$ coincide. On the other hand, when the internal resistance is very high, to avoid the losses due to high discharge rate, it is optimal to transmit over the entire frame duration, i.e., $\tau_4=T$ is optimal. In such a case, the point $A$ coincides with point $B$ and the point $C$ coincides with $D$\footnote{In the ideal case, when both the circuit cost and the internal resistance are zero, we will have the similar scenario.}. However, when the internal resistance is not sufficiently high, we may have $0<\tau_4<T$. Hence, the points $A$ ($C$) and $B$ ($D$) may not coincide. Now, to achieve the $CD$ curve, we need to decode user $2$ first and cancel the interference on user $1$. Hence, the rate achieved by user $1$ is  $R^{(1)}=\tau_2\log(1+{E}_2^{(1)}/\tau_2)+\tau_4\log(1+{E}_4^{(1)}/\tau_4)$. 
Let $R^{(1),0}$ be the rate in user $1$ for which we would like to find the rate in user $2$ on the $CD$ curve. This rate can  be obtained by solving the following convex optimization problem. 
\begin{subequations}\label{eq:p_MAC3}
	\begin{align}
	\tilde{R}=\underset{\substack{{E}_i^{(u)},{e}_i^{(u)},\tau_1}}{\text{maximize}} &\;\;R^{(1)}+R^{(2)}\\
	\text{subject to}&\;\;\eqref{eq:p_nonEH1_c1}-\eqref{eq:p_nonEH1_c5}\\
	&\;\; R^{(1)}\geq R^{(1),0} \label{eq:CD}
	\end{align}
\end{subequations}
for $i=1,\ldots,4$ and $u=1,2$. The  maximum sum-rate does not increase with  $R^{(1)}$ in $[R^{(1),C},C^{(1)}]$. Hence, the maximum rate in user $2$ when user $1$ transmits at $R^{(1),0}\in [R^{(1),C},C^{(1)}] $ is given by $\tilde{R}-R^{(1),0}$.  Along the above lines, we can achieve the rates on the curve $AB$.

\section{Multiple Access Channel with $U\geq 1$ Users}\label{sec:Uusers}
In this section, we generalize the results in the previous sections for the case when the number of users is arbitrary. 
As in the two-user case, it may not be optimal for all the users to transmit over the entire time duration of $T$ seconds. Hence, we first obtain the optimal frame structure. 
\subsubsection{Optimal Frame Structure}
The number of users transmitting  at any point in time can be $0,\ldots,U$. Note that there are $2^U$ such combinations given by the set of all subsets (power set) of $\mathcal{U}$, represented by $\mathcal{P}(\mathcal{U})$. Now, in order to exploit all the available degrees of freedom, we divide the total available duration of $T$ seconds into $2^U$ phases. We then order the elements in  $\mathcal{P}(\mathcal{U})$ in any manner and represent the  $i$th element (which is a set) in $\mathcal{P}(\mathcal{U})$ by $\mathcal{U}_i$. The users in $\mathcal{U}_i$ transmit in phase $i$. We assume that the length of phase $i$ is $\tau_i$ seconds. If $\mathcal{U}_i$ contains more than one users, the signals from all the users are superimposed and transmitted. The information is decoded by successive interference cancellation at the receiver. 
\subsubsection{Maximum Achievable Rate Region}
Let $E^{(u)}_i$ and $e^{(u)}_i$ denote the total transmit energy and the energy drawn from the battery in user $u\in \mathcal{U}_i$ in phase $i\in\{1,\ldots,2^U\}$, respectively.
Based on the optimal frame structure,  as in \eqref{eq:mac_c1} for the two-user case, we apply the following constraints. 
\begin{align}\label{eq:macgen_c1}
	E^{(u)}_i,e^{(u)}_i,\gamma^{(u)}_i\triangleq 0, \quad \forall\; u\notin \mathcal{U}_i,\; i\in \mathcal{N}
\end{align}
where $\mathcal{N}=\{1,2,\ldots,2^U\}$.  
Now, generalizing the results in the two-frame case along the lines in the  proof of Proposition 1 in \cite{OnlineOzgur}, for an arbitrary $U$ and  given $\tau_i$'s and  $E_i^{(u)}$'s, the maximum achievable rate region is given by 
\begin{align}\label{eq:genCapRegion}
	&\mathcal{R}\left(E^{(u)}_i,\tau_i,u\in\mathcal{U},i\in\mathcal{N}\right)=\left\{R^{(u)}:\sum_{u\in \mathcal{S}}R^{(u)}\leq \sum_{i=1}^{2^U} \tau_i \log\left(1+
	\frac{\sum_{u\in \mathcal{S}} E^{(u)}_i}{\tau_i}\right),\forall \mathcal{S}\subseteq \mathcal{U} \right\} &&
\end{align}
Finally,  the maximum achievable rate region of the MAC can be obtained by taking the convex hull of the union of the maximum achievable rate regions in \eqref{eq:genCapRegion} over all feasible  $\{E^{(u)}_i,\tau_i,u\in \mathcal{U},i\in \mathcal{N}\}$, i.e.,
\begin{align}
	\mathcal{C}=\bigcup_{\{E_i^{(u)},\tau_i,i\in \mathcal{N},u\in\mathcal{U}\}}\mathcal{R}\left(E_i^{(u)},\tau_i, u\in \mathcal{U},i\in \mathcal{N}\right)
\end{align}
subject to \eqref{eq:mac_c2}-\eqref{eq:mac_c5} and  \eqref{eq:macgen_c1}, where $\mathcal{R}(\cdot)$ is defined in \eqref{eq:genCapRegion}. 

\subsubsection{Maximum Sum-Rate}
From \eqref{eq:genCapRegion}, to maximize the sum-rate when the number of users is arbitrary, we need to solve the following optimization problem. 
\begin{subequations}\label{eq:p_MACU}
	\begin{align}
		\mathrm{(P6)}:\underset{\substack{E_i^{(u)},e_i^{(u)},\tau_i\\ i\in \mathcal{N},u\in \mathcal{U}}}{\text{maximize}} &\;\; \sum_{i=1}^{2^U}\tau_i\log\left(1+\frac{\sum_{u=1}^{U}E^{(u)}_i}{\tau_i}\right)\;\;\;&&\\
		\text{subject to} & \;\; \eqref{eq:p_nonEH1_c1}-\eqref{eq:p_nonEH1_c4},  \eqref{eq:macgen_c1} &&
	\end{align}
\end{subequations}
As in the two-user case, noting that perspective preserves convexity, $\mathrm{(P6)}$ in \eqref{eq:p_MACU} is convex and we solve it numerically. 

We now generalize Lemma \ref{thm:sum-rate} for $U$ users. Let $\mathcal{I}_n$ be the set of phase indices where some combination of $n\in\{0,\ldots,U\}$ users transmit. Then, as in the two-user case, for NOMA, $\tau_i\triangleq 0, \; \forall\; i\in \{\mathcal{I}_0, \ldots,\mathcal{I}_{U-1}\}$ and for TDMA,  $\tau_i\triangleq 0, \; \forall\; i\in \{\mathcal{I}_2, \ldots,\mathcal{I}_U\}$. Then, we have the following theorem. 
\begin{theorem}\label{thm:sum-rateUusers}
Lemma \ref{thm:sum-rate} generalizes to any $U\geq 2$.
\end{theorem}
\begin{proof}
	See Appendix D. 
\end{proof}
From the above theorem, we note that  due to the non-linear discharging function and the circuit cost, simpler strategies,  such as the TDMA and NOMA,  are no longer optimal.
Further, the number of optimization variables increase exponentially in $U$ and hence,  the complexity in solving $\mathrm{(P6)}$ in \eqref{eq:p_MACU} is exponential in $U$. 
We now make the following observation on the structure of the optimal solution to $\mathrm{(P6)}$ in \eqref{eq:p_MACU}. 
\begin{proposition}\label{thm:uC2}
In the optimal solution to $\mathrm{(P6)}$ in \eqref{eq:p_MACU}, we have, $\tau_i^*\geq 0$ for all $i\in \{\mathcal{I}_n,\mathcal{I}_{n+1}\}$ and $\tau_i^*=0$ for all $i\notin \{\mathcal{I}_n,\mathcal{I}_{n+1}\}$  for some $n\in \{0,1,\ldots,U-1\}$.
\end{proposition}
\begin{proof}
	See Appendix E. 
\end{proof}
The above proposition is intuitive. For the two-user case, the proposition implies that when the users are superimposed, there does not exist any time instant where  no transmission takes place, in the optimal solution. That is,  the users must first occupy the entire frame duration with individual transmissions, before overlapping their transmissions. A similar intuition can be applied to more than two users. 


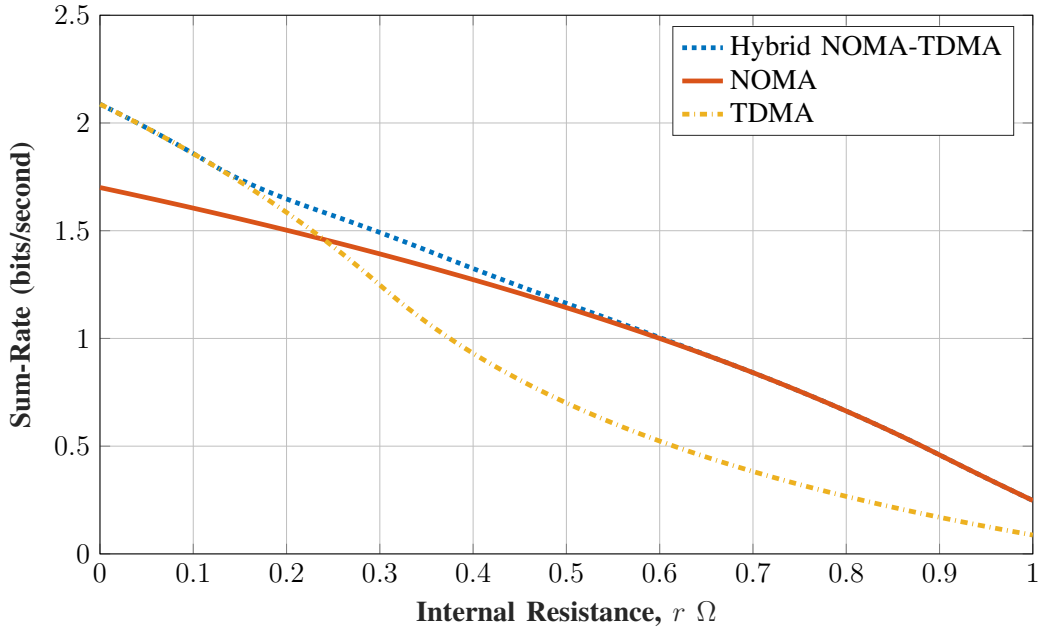
\begin{figure}[t]
	\centering
%
%
\definecolor{mycolor1}{rgb}{0.00000,0.44700,0.74100}%
\definecolor{mycolor2}{rgb}{0.85000,0.32500,0.09800}%
\definecolor{mycolor3}{rgb}{0.92900,0.69400,0.12500}%
\begin{tikzpicture}[scale=.9]

\begin{axis}[%
width=5.425in,
height=3.135in,
at={(0.91in,0.565in)},
scale only axis,
xmin=0,
xmax=1,
xlabel style={font=\bfseries\color{white!15!black}},
xlabel={Internal Resistance, $r$ \si{\ohm}},
ymin=0,
ymax=2.5,
ylabel style={font=\bfseries\color{white!15!black}},
ylabel={Sum-Rate  (bits/second)},
axis background/.style={fill=white},
xmajorgrids,
ymajorgrids,
legend style={legend cell align=left, align=left, draw=white!15!black}
]
\addplot [color=mycolor1, line width=2.0pt,dotted]
  table[row sep=crcr]{%
0	2.08746283447698\\
0.00999999999999979	2.06608917261522\\
0.02	2.04439409755097\\
0.0299999999999998	2.02236778687079\\
0.04	1.99999996719753\\
0.0499999999999998	1.97727987983804\\
0.0600000000000001	1.9541963031066\\
0.0699999999999998	1.93073733038084\\
0.0800000000000001	1.90689058635571\\
0.0899999999999999	1.8826430266344\\
0.1	1.85798098920362\\
0.11	1.83289000874348\\
0.12	1.80735488021961\\
0.13	1.78255837683578\\
0.14	1.75968199456071\\
0.15	1.73837448538365\\
0.16	1.71835717967986\\
0.17	1.69940788807907\\
0.18	1.68134725253074\\
0.19	1.66402881605365\\
0.2	1.64733180789951\\
0.21	1.63115569500263\\
0.22	1.615416272488\\
0.23	1.60004224582667\\
0.24	1.58496245360892\\
0.25	1.56985559429882\\
0.26	1.55458884863371\\
0.27	1.53915877514575\\
0.28	1.5235619385763\\
0.29	1.50779463379963\\
0.3	1.49185303812335\\
0.31	1.47573341645083\\
0.32	1.45943161534461\\
0.33	1.44294348805778\\
0.34	1.42626473258533\\
0.35	1.40939092686745\\
0.36	1.39231740519686\\
0.37	1.37510019986577\\
0.38	1.35803915784726\\
0.39	1.34115860832255\\
0.4	1.32444005023246\\
0.41	1.30786629481124\\
0.42	1.29142140982004\\
0.43	1.27509058788233\\
0.44	1.25886006035283\\
0.45	1.24271687464022\\
0.46	1.22664897372103\\
0.47	1.21064502618403\\
0.48	1.19469436932176\\
0.49	1.17878699477706\\
0.5	1.16291344007773\\
0.51	1.14706469249028\\
0.52	1.1312322826405\\
0.53	1.11540817998979\\
0.54	1.09958467021707\\
0.55	1.08375446156753\\
0.56	1.06791058743965\\
0.57	1.05204635991631\\
0.58	1.03615538988643\\
0.59	1.02023153404807\\
0.6	1.00426890405172\\
0.61	0.988261809445712\\
0.62	0.972204743352506\\
0.63	0.956092457121421\\
0.64	0.939919772789819\\
0.65	0.923681726340524\\
0.66	0.907373494080581\\
0.67	0.8909903558175\\
0.68	0.874527779174267\\
0.69	0.8579812359185\\
0.7	0.84130224246561\\
0.71	0.824428418586807\\
0.72	0.807354907528754\\
0.73	0.790076916577355\\
0.74	0.772589487275064\\
0.75	0.754887484969777\\
0.76	0.736965575937853\\
0.77	0.718818213759429\\
0.78	0.700439708411021\\
0.79	0.681824035116115\\
0.8	0.662965009813822\\
0.81	0.643856188778241\\
0.82	0.624490855168324\\
0.83	0.604862056086659\\
0.84	0.584962493615285\\
0.85	0.56478461442364\\
0.86	0.544320514792551\\
0.87	0.523561951182219\\
0.88	0.502500336236393\\
0.89	0.48112668184698\\
0.9	0.459431593695896\\
0.91	0.437649194891593\\
0.92	0.416017253559328\\
0.93	0.394531840587462\\
0.94	0.373189106683223\\
0.95	0.351985324750038\\
0.96	0.33091687344723\\
0.97	0.309980237149885\\
0.98	0.289171997741001\\
0.99	0.268488831818474\\
1	0.247927510325278\\
};
\addlegendentry{Hybrid NOMA-TDMA}

\addplot [color=mycolor2, line width=2.0pt]
  table[row sep=crcr]{%
0	1.70043971814109\\
0.03	1.6724253419715\\
0.0600000000000001	1.64385618977472\\
0.0900000000000001	1.61470984411521\\
0.12	1.58496250072116\\
0.15	1.55458885167764\\
0.18	1.52356195605701\\
0.21	1.49185309632967\\
0.24	1.4594316186373\\
0.27	1.4262647547021\\
0.3	1.39231742277876\\
0.33	1.35755200461808\\
0.36	1.32192809488736\\
0.39	1.28540221886225\\
0.42	1.24792751344359\\
0.45	1.20945336562895\\
0.48	1.16992500144231\\
0.51	1.12928301694497\\
0.54	1.08746284125034\\
0.56	1.05889368905357\\
0.58	1.02974734339405\\
0.6	1\\
0.62	0.969626350956481\\
0.64	0.938599455335857\\
0.66	0.906890595608518\\
0.68	0.874469117916141\\
0.7	0.841302253980942\\
0.72	0.807354922057604\\
0.74	0.772589503896927\\
0.76	0.736965594166206\\
0.78	0.700439718141092\\
0.8	0.662965012722429\\
0.82	0.624490864907794\\
0.84	0.584962500721156\\
0.86	0.544320516223811\\
0.88	0.502500340529183\\
0.9	0.459431618637297\\
0.93	0.394531843844202\\
0.96	0.330916878114617\\
0.99	0.268488835925902\\
1	0.247927513443585\\
};
\addlegendentry{NOMA}

\addplot [color=mycolor3, dashdotted, line width=2.0pt]
  table[row sep=crcr]{%
0	2.08746280367061\\
0.02	2.04439410936358\\
0.04	1.99999999272152\\
0.0600000000000001	1.95419630221903\\
0.0800000000000001	1.90689058666929\\
0.1	1.8579809873683\\
0.12	1.80735491687175\\
0.14	1.75488748078842\\
0.16	1.70043969262274\\
0.18	1.6438561848936\\
0.2	1.58496249895663\\
0.22	1.52356195469214\\
0.24	1.45943161553749\\
0.26	1.39231740714735\\
0.28	1.32192809209725\\
0.29	1.28540220983092\\
0.31	1.21071060684157\\
0.32	1.17492567547036\\
0.33	1.1404812044843\\
0.34	1.10729399433694\\
0.35	1.07528811157737\\
0.36	1.04439410074127\\
0.37	1.01454830741879\\
0.38	0.985692158288131\\
0.39	0.95777175556403\\
0.4	0.930737331513265\\
0.41	0.904542832700758\\
0.42	0.879145594122981\\
0.43	0.854506013879863\\
0.44	0.830587226288217\\
0.45	0.807354917916476\\
0.46	0.784777071613338\\
0.47	0.762823759220363\\
0.48	0.741466983183167\\
0.5	0.700439713641217\\
0.52	0.661504053137162\\
0.54	0.624490856522264\\
0.56	0.589249849374414\\
0.58	0.555646923355163\\
0.6	0.523561943826478\\
0.62	0.492886901777707\\
0.64	0.463524357420866\\
0.66	0.435386127246013\\
0.68	0.408392180401572\\
0.7	0.382469631380979\\
0.72	0.357551998886686\\
0.74	0.333578439833417\\
0.76	0.310493185917605\\
0.78	0.28824495254959\\
0.8	0.266786524634965\\
0.82	0.246074325247382\\
0.84	0.226068065195554\\
0.86	0.206730431672287\\
0.88	0.188026799459392\\
0.9	0.169924992207286\\
0.92	0.152395067503963\\
0.94	0.135409109928071\\
0.96	0.118941062645755\\
0.98	0.102966568762255\\
1	0.0874628299373721\\
};
\addlegendentry{TDMA}

\end{axis}
\end{tikzpicture}%
	\caption{Comparison of maximum sum-rates in the three-user Gaussian MAC  for $T=1$ \si{\second}, $B^{(1)},B^{(2)},B^{(3)}=1.25$ \si{\joule}, $\gamma^{(1)},\gamma^{(2)},\gamma^{(3)}=0.5$ \si{\watt} and $r_1,r_2,r_3=r$ \si{\ohm}.  }
	\label{fig:sum-Num}
\end{figure}

\section{Numerical Results}\label{sec:numerical results}
We now obtain numerical results. Based on \cite{IR,Krieger}, we assume $g^{(u)}(d)=-0.44r^{(u)}d^2+d$, where $r^{(u)}$ \si{\ohm} is the internal resistance of the battery in user $u$. 

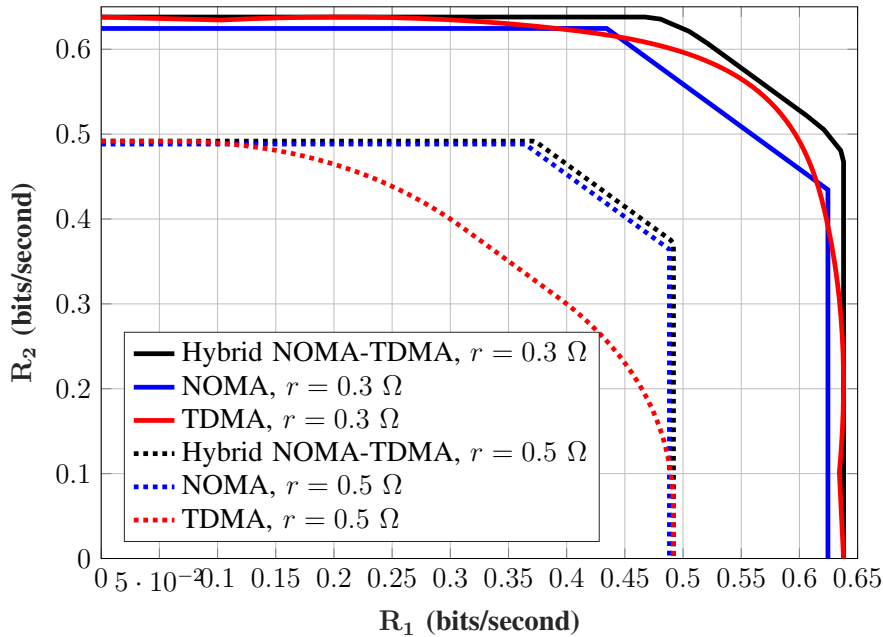
\begin{figure}[t]
	\centering
%
%
\begin{tikzpicture}[scale=.9]

\begin{axis}[%
width=4.399in,
height=3.211in,
at={(0.88in,0.835in)},
scale only axis,
xmin=0,
xmax=0.65,
xlabel style={font=\bfseries\color{white!15!black}},
xlabel={$\mathbf{R_1}$ (bits/second)},
ymin=0,
ymax=0.65,
ylabel style={font=\bfseries\color{white!15!black}},
ylabel={$\mathbf{R_2}$ (bits/second)},
axis background/.style={fill=white},
xmajorgrids,
ymajorgrids,
legend style={at={(0.03,0.03)}, anchor=south west, legend cell align=left, align=left, draw=white!15!black}
]
\addplot [color=black, line width=2.0pt]
  table[row sep=crcr]{%
0	0.637947011203481\\
0.467450863173795	0.637947011205043\\
0.480688609145559	0.635474099853918\\
0.50524947465977	0.621047157882205\\
0.52130278657785	0.606620198907522\\
0.606643735433369	0.521279241653818\\
0.621047157882205	0.50524947465977\\
0.635474099853918	0.480688609145559\\
0.637947011203481	0.467422698493232\\
0.637947011203481	0\\
};
\addlegendentry{Hybrid NOMA-TDMA, $r=0.3$ \si{\ohm}}

\addplot [color=blue, line width=2.0pt]
  table[row sep=crcr]{%
0	0.624490864907794\\
0.434402824145775	0.624490864907794\\
0.624490864907794	0.434402824145775\\
0.624490864907794	0\\
};
\addlegendentry{NOMA, $r=0.3$ \si{\ohm}}

\addplot [color=red, line width=2.0pt]
  table[row sep=crcr]{%
0	0.637947027544129\\
0.103328398113129	0.634437196840895\\
0.128162155630656	0.635905018973862\\
0.152496742142806	0.636970515562048\\
0.176456950400923	0.637645747377566\\
0.199917987653662	0.637934146848028\\
0.223004646652369	0.637844499274603\\
0.245592134645698	0.63738491144854\\
0.267805244384994	0.636560785336641\\
0.289643975870256	0.635377240322828\\
0.311108329101486	0.633839643290894\\
0.332198304078682	0.631953615269364\\
0.352913900801846	0.629725038136775\\
0.373255119270977	0.627160061391359\\
0.393221959486074	0.624265108991785\\
0.412814421447138	0.621046886278851\\
0.428320259797345	0.618159772584467\\
0.441505898898949	0.61531891843672\\
0.453830772920218	0.612283590002324\\
0.465274499175819	0.609082463037281\\
0.475920057739801	0.605720612525882\\
0.485928199483888	0.602169671929401\\
0.495262557276655	0.598464504718265\\
0.504059291455079	0.594572803679238\\
0.512355673885245	0.590493773852207\\
0.520185042770855	0.58622655223248\\
0.527577399940287	0.581770199937064\\
0.534559893375013	0.577123693464216\\
0.541221001207672	0.572236810301568\\
0.547513665176554	0.567152537125077\\
0.553517060621417	0.561815223843413\\
0.55924260086194	0.556215324246396\\
0.564700879659197	0.550342595593249\\
0.569901758007523	0.544186017580736\\
0.574854439116604	0.537733698299985\\
0.579567533498914	0.53097276452657\\
0.584092356174141	0.523817436291852\\
0.588429889810985	0.516242086651201\\
0.592581032831277	0.508218476306478\\
0.596509886866184	0.499798284645263\\
0.60029272685871	0.490784973330461\\
0.603858773528446	0.481310187829539\\
0.60727592045106	0.471151389590625\\
0.610510099215312	0.460345621248078\\
0.613561861416636	0.448835363265823\\
0.616456635765027	0.436440059111238\\
0.619189355870688	0.423060511960444\\
0.621924216144695	0.407697918642019\\
0.625043357206294	0.388105456680955\\
0.62784116586839	0.368138616465858\\
0.630310901872387	0.347797397996727\\
0.632434363975659	0.327206594024908\\
0.634221544895384	0.306241411799056\\
0.635666400863401	0.284901851319171\\
0.636768393466476	0.263063119833908\\
0.637512580331997	0.240850010094612\\
0.637893813425448	0.218262522101283\\
0.637907186132837	0.195300655853921\\
0.637545120061874	0.171839618601182\\
0.636802200245904	0.148004203094409\\
0.635667556561492	0.123669616582259\\
0.634319818499041	0.101581299594308\\
0.637947027544129	0\\
};
\addlegendentry{TDMA, $r=0.3$ \si{\ohm}}


\addplot [color=black, line width=2.0pt,dotted]
table[row sep=crcr]{%
	0	0.491760690819527\\
	0.370822697142469	0.491760690847686\\
	0.37502987439648	0.489490775389741\\
	0.489496102302511	0.375024536985232\\
	0.491760690819527	0.3708154944072\\
	0.491760690819527	0\\
};
\addlegendentry{Hybrid NOMA-TDMA, $r=0.5$ \si{\ohm}}

\addplot [color=blue, line width=2.0pt,dotted]
table[row sep=crcr]{%
	0	0.488286481309482\\
	0.36415633027666	0.488286481309482\\
	0.488286481309482	0.36415633027666\\
	0.488286481309482	0\\
};
\addlegendentry{NOMA, $r=0.5$ \si{\ohm}}

\addplot [color=red, line width=2.0pt,dotted]
table[row sep=crcr]{%
	0	0.491760697319601\\
	0.0697637959268528	0.491575020620986\\
	0.0823016668815783	0.491020013643162\\
	0.0946294059208616	0.490099618175641\\
	0.106747013044702	0.488817942965195\\
	0.118584444281287	0.48719003622996\\
	0.130211743602429	0.48521419134732\\
	0.141628911008129	0.48289520718989\\
	0.152835946498386	0.480238073016389\\
	0.163902894045016	0.477227677505221\\
	0.174759709676203	0.473884627267798\\
	0.185406393391947	0.470214323397595\\
	0.195912989164064	0.466194220949255\\
	0.206209453020738	0.461852768973756\\
	0.216365828933783	0.457161998090186\\
	0.226382116903201	0.452119105003843\\
	0.236188272957176	0.446760764895088\\
	0.245854341067523	0.441049970695419\\
	0.255380321234242	0.434983108251458\\
	0.264766213457333	0.428556097182288\\
	0.274012017736795	0.421764328218044\\
	0.283047690100815	0.414659543020675\\
	0.291943274521207	0.407186628759831\\
	0.300698770997971	0.399339453751973\\
	0.309314179531106	0.391111066658463\\
	0.328996535610871	0.371443182530221\\
	0.397380103159912	0.302800090152394\\
	0.405348946890465	0.294044593675631\\
	0.412938145453114	0.285149009255239\\
	0.420154262104217	0.276113336891219\\
	0.427003091662278	0.26693757658357\\
	0.433489753237089	0.257621728332294\\
	0.439662520721752	0.248095748165575\\
	0.445475256187841	0.238429680055228\\
	0.450931845011325	0.228623524001252\\
	0.456035734253498	0.218677280003649\\
	0.460789984338932	0.208590948062417\\
	0.465226112721994	0.198294484205743\\
	0.469312340006045	0.187857932405441\\
	0.473074525569787	0.177211248689696\\
	0.476486278018396	0.166424477030324\\
	0.479568161817826	0.155427573455508\\
	0.482298905602575	0.144290581937065\\
	0.484694103726505	0.132943458503179\\
	0.486748672189689	0.121386203153851\\
	0.488457720717552	0.109618815889081\\
	0.489816547730663	0.0976412967088683\\
	0.490820635664679	0.0854536456132132\\
	0.491465646575639	0.0730558626021159\\
	0.491747956769903	0.0603779037037622\\
	0.491760697319601	0.0570157930566849\\
	0.491760697319601	0\\
};
\addlegendentry{TDMA, $r=0.5$ \si{\ohm}}

\end{axis}
\end{tikzpicture}%
	\caption{Comparison of the maximum achievable rate regions in the two-user Gaussian MAC  for $T=1$ \si{\second}, $B^{(1)},B^{(2)}=1.25$ \si{\joule}, $\gamma^{(1)},\gamma^{(2)}=0.5$ \si{\watt} and $r_1,r_2=r$ \si{\ohm}. }
	\label{fig:MAC-Num}
\end{figure}
\subsection{Circuit cost, Internal Resistance and Sum-rate}
In Fig. \ref{fig:sum-Num}, we study the impact of the circuit cost and internal resistance on the maximum sum-rate. From Fig. \ref{fig:sum-Num}, we first note, unlike in the ideal battery case  where TDMA always achieves the maximum sum-rate\cite{GP},  when the internal resistance is non-zero, the performance of TDMA drops significantly with the internal resistance. Second, when the internal resistance is low, the performance of TDMA  is very close to the performance of NOMA-TDMA. On the other hand, when the internal resistance is high, the performance of  NOMA is close to the performance of the hybrid NOMA-TDMA. Note that the hybrid NOMA-TDMA takes the features from both the strategies and hence, its performance is the best.

\subsection{Achievable Rate Regions}
We plot the maximum achievable rate-regions for $r=0.3$ \si{\ohm} and $r=0.5$ \si{\ohm} in Fig. \ref{fig:MAC-Num}. We see that the maximum sum-rate in the TDMA strategy is higher than the NOMA strategy when $r=0.3$ \si{\ohm}. However, for $r=0.5$ \si{\ohm}, the maximum sum-rate in the TDMA strategy is significantly lower than the  NOMA strategy. Further, the rate region achieved by the hybrid NOMA-TDMA strategy is larger than other two strategies for both the values of the internal  resistance.   

\section{Conclusions}\label{sec:reflections}
In this work, we studied the impact of the circuit cost and internal resistance of the battery on sum-rates and achievable rate regions of a multiple access channel.  When the internal resistance of the battery is non-zero, we have shown that TDMA may not achieve the maximum sum-rate.  Further, when the circuit cost is non-zero, NOMA does not achieve the maximum sum-rate. Finally, we have shown that a hybrid NOMA-TDMA  
achieves the maximum sum-rate in general. We have also shown that   
the maximum achievable rate regions in NOMA and TDMA  are contained within the maximum achievable rate region of the hybrid NOMA-TDMA, numerically.

\section*{Appendix}

\subsection{Proof of Lemma \ref{lemma:p2p0}}
We first note that it is optimal to utilize all the energy stored in the battery,  subject to the maximum discharge rate constraint, $d\leq D_0$. 
Hence, both \eqref{eq:p2p-1-c1} and \eqref{eq:p2p-1-c2} must be satisfied with equality, i.e., $d=\min(B/\tau,D_0)$ and $P=\left[g\left(\min({B/\tau,D_0})\right)-\gamma\right]^+$, where $[x]^+=\max(x,0)$. 
Now,  $\mathrm{(P1)}$ in \eqref{eq:p2p-1} can be reformulated as:
\begin{align}\label{eq:p2p-2}
&\mathrm{(P1')}:\underset{\substack{0\leq \tau \leq T}}{\text{maximize}} \;\tau\log\left(1+g\left(\min\left(\frac{B}{\tau},D_0\right)\right)-\gamma\right)&&
\end{align}
Note that the objective function of $\mathrm{(P1')}$ in \eqref{eq:p2p-2} is a monotonically increasing continuous function of $\tau$ over $\tau \in [0,T]$ -- it is a linearly increasing function over $\tau\in [0,B/D_0]$, as $\min(B/\tau,D_0)=D_0$, and a strictly concave increasing  function over $\tau\in (B/D_0,T)$.  Clearly, for  $\tau\in [0,B/D_0]$, the objective function attains its maximum at $\tau=B/D_0$. 
Hence, we can obtain the optimal solution by solving the following convex optimization problem in \eqref{eq:p2p-3}.  Hence, the proof. 

\subsection{Proof of Proposition \ref{lemma:p2p}}
The Lagrangian of $\mathrm{(P2)}$ in \eqref{eq:p2p-3} is given by, 
\begin{align}
&L=\tau\log\left(1+g\left( \frac{B}{\tau}\right)-\gamma\right)+\lambda(\frac{B}{D_0}-\tau)+\mu(\tau-T)&&
\end{align}
where $\lambda,\mu\geq 0$ are non-negative Lagrange multipliers. Differentiating $L$ with respect to $\tau$ and equating to zero, we get the following stationarity condition that gets satisfied for the optimal $\tau^*$, 
\begin{align}
&\frac{-\frac{B}{\tau}g'(\frac{B}{\tau})}{1+g\left( \frac{B}{\tau}\right)-\gamma}+
\log\left(1+g\left( \frac{B}{\tau}\right)-\gamma\right)-\lambda+\mu=0&&
\end{align}
Further, due to complementary slackness conditions, we must have, $\lambda({B}/{D_0}-\tau)=0,\;\mu(\tau-T)=0$ in the optimal solution. Hence, when $B/D_0<\tau<T$, from the above KKT conditions,  the optimal $\tau$ must satisfy the following equation. 
\begin{align}\label{eq:opt_tau}
&\left(1+g\left( \frac{B}{\tau}\right)-\gamma\right)\log\left(1+g\left( \frac{B}{\tau}\right)-\gamma\right)=\frac{B}{\tau}g'\left(\frac{B}{\tau}\right)&&
\end{align}
Suppose the optimal transmission duration, $\tau^*$ when $B=B'$ is $\tau'$. Noting that \eqref{eq:opt_tau} depends only on the ratio, $B/\tau$, clearly, for a given $B$, any $\tau$ that satisfies   $B/\tau=B'/\tau'$ is a solution to \eqref{eq:opt_tau}.
Since the objective function is strictly concave over  $\tau\in (B/D_0,T)$, we conclude that $\tau^*=B/(B'/\tau')$ is the optimal solution.   
Hence, clearly, $\tau^*$ linearly increases with $B$. Hence, the proof. 
\subsection{Proof of Lemma \ref{thm:sum-rate}}
We construct examples to prove each of the above statements. Assume $B^{(1)}=B^{(2)}$, $g^{(1)}(\cdot)=g^{(2)}(\cdot)$ and $\gamma^{(1)}=\gamma^{(2)}=0$. For simplicity, assume that $2B^{(1)}/T\leq D^{(1)}_0=D^{(2)}_0$. Since $\gamma^{(u)}=0$, 
there is no loss in the circuit operation and we can transmit for the entire frame duration, i.e., $\tau_2+\tau_3+\tau_4=T$ and $\tau_1=0$ in the optimal solution. 
\begin{enumerate}
	\item In the above setting, from \eqref{eq:cap_regionS12}, we get, $R_{\mathrm{NOMA}}=T\log\left(1+2g^{(1)}\left({B^{(1)}}/{T}\right)\right)$. As the parameters of both the users are identical,  in the TDMA strategy in \eqref{eq:mac-II},  the maximum value of $R_{\mathrm{TDMA}}$ is achieved at $\tau_2=\tau_3=0.5$. Hence, $R_{\mathrm{TDMA}}=T\log\left(1+g^{(1)}\left({2B^{(1)}}/{T}\right)\right)$. Due to the strict concavity of $g^{(1)}(\cdot)$ function, we have, $g^{(1)}\left({2B^{(1)}}/{T}\right)< 2g^{(1)}\left({B^{(1)}}/{T}\right)$. Hence, $R_{\mathrm{TDMA}}< R_{\mathrm{NOMA}}$.
	\item Since the{ hybrid NOMA-TDMA} strategy is a generalization of the strategy in \eqref{eq:cap_regionS12}, we have, $R_{\mathrm{NOMA-TDMA}}\geq R_{\mathrm{NOMA}}$. Now, from the first result, we have,  $R_{\mathrm{TDMA}}< R_{\mathrm{NOMA}} \leq R_{\mathrm{NOMA-TDMA}}$. This shows that the {hybrid NOMA-TDMA} strategy achieves a strictly higher rate than the TDMA strategy. 
\end{enumerate}
Hence, the proof. 

\subsection{Proof of Theorem \ref{thm:sum-rateUusers}}
This theorem is proved along the lines in the proof of Lemma \ref{thm:sum-rate},  by constructing appropriate examples as described below. 
Assume $B^{(u)}=B$, $g^{(u)}(\cdot)=g(\cdot)$ and $\gamma^{(u)}=0$ for all $u\in \mathcal{U}$. In this example, from \eqref{eq:cap_regionS12}, we get, $R_{\mathrm{NOMA}}=T\log\left(1+Ug\left({B}/{T}\right)\right)$. As the parameters are identical for both the users,  in the TDMA strategy, the maximum value of $R_{\mathrm{TDMA}}$ is achieved when all the users transmit for the equal duration of $T/U$ seconds. Hence, $R_{\mathrm{TDMA}}=T\log\left(1+g\left({UB}/{T}\right)\right)$. Due to the strict concavity of $g(\cdot)$ function, we have, $g\left({UB}/{T}\right)< Ug\left({B}/{T}\right)$. Hence, $R_{\mathrm{TDMA}}< R_{\mathrm{NOMA}}$. From the second part of Lemma \ref{thm:sum-rate}, the second statement is trivial.  Hence, the proof. 

\subsection{Proof of Proposition \ref{thm:uC2}}
To prove the proposition, we suppose the users incrementally use the available energy in their batteries in steps of $\Delta B$, where  $\Delta B$ can be arbitrarily small. 
Now, even when there is only one user, due to the non-zero circuit cost, it is not optimal to transmit over the entire frame duration. Let $\tau^{(u)*}$ be the single user optimal transmit duration in user $u$,  obtained by solving $\mathrm{(P2)}$ in \eqref{eq:p2p-3} with energy of $\Delta B$ units. From Proposition \ref{lemma:p2p}, $\tau^{(u)*}$ is a linearly increasing function of $\Delta B$ and, clearly, $\tau^{(u)*}=0$ when $\Delta B=0$. Hence, when $\Delta B$ is sufficiently small, we will have, $\sum_{u=1}^{U}\tau^{(u)*}<T$. Clearly, in this case, each user must transmit exactly the same way as the single user case, i.e., the users need not overlap their transmissions for optimality. Hence, $\tau_i=0,\; \forall i\notin \{\mathcal{I}_0,\mathcal{I}_1\}$. 
Now, consider the situation when the users allocate additional $\Delta B$ units of energy. In this case, based on Proposition  \ref{lemma:p2p}, the optimal $\tau^{(u)*}$'s increase. As we continue to increment the total amount of energy allocated by the users for the transmission, at some stage, we will have, $\sum_{u=1}^{U}\tau^{(u)*}=T$. This implies, $\tau_i=0,\; \forall i\notin \{\mathcal{I}_1\}$.  
Now, in order to allocate an additional amount of energy, 
the users may either allocate more power to those time windows which have been  already occupied by themselves. In this case, since the discharge function is concave and as there is an upper limit on the discharge power of the battery, after a certain level of the transmit power, it may no longer be optimal to increase the transmit power. In this case, the user can increase the duration over which it transmits, either at the cost of reducing the time duration allocated to other users, in which case the transmit powers of the other users increase, or by superimposing on the time windows where other users have already allocated the power. 
When any two users superimpose, the total transmit power in the time window over which they superimpose is higher than the transmit power over the remaining fraction of the frame duration.  Hence, it is suboptimal for any of the remaining users to transmit over the same window over which the two users are transmitting. As a result, the remaining users occupy the time windows where only a single user is transmitting, i.e., we will have,  $\tau_i=0,\; \forall i\notin \{\mathcal{I}_1,\mathcal{I}_2\}$. Now, as the users continue to allocate more amount of energy, at some stage, there will be a superposition of a set of two users at any given instant of time. That is, we will have, $\tau_i=0,\; \forall i\notin \{\mathcal{I}_2\}$. 
Proceeding along the above lines, one can find that as the users increment the total amount of energy, we will have, $\tau_i=0,\; \forall i\notin \{\mathcal{I}_2,\mathcal{I}_3\}$ and so on, until all the users superimpose their transmissions for the entire $T$ seconds. Hence, the proof. 
\bibliographystyle{ieeetran}
\bibliography{IEEEabrv,isit}

\end{document}